\documentclass[11pt]{amsart}
 
\usepackage{microtype}
\usepackage{enumerate}
\usepackage{hyperref}
\usepackage{amsmath,amssymb,amsthm}
\usepackage{graphicx}

\hypersetup{
    pdftitle=   {Maximum Matching Width},
   pdfauthor=  {Jisu Jeong, Sigve Hortemo S\ae ther, Jan Arne Telle}
}

\newcommand\abs[1]{\lvert #1\rvert}
\newcommand{\compl}[1]{\overline{#1}}
\newcommand{\mm}{\operatorname{mm}}
\newcommand{\bw}{\operatorname{bw}}
\newcommand{\tw}{\operatorname{tw}}
\newcommand{\mmw}{\operatorname{mmw}}
\newcommand{\nat}{\mathbb{N}}

\newtheorem{theorem}{Theorem}[section]
\newtheorem{lemma}[theorem]{Lemma}
\newtheorem{corollary}[theorem]{Corollary}

\begin{document}

\title[Maximum Matching Width]{Maximum matching width: new characterizations and a fast algorithm for dominating set}
\author{Jisu Jeong}
\author{Sigve Hortemo S\ae ther}
\author{Jan Arne Telle}
\address[Jeong]{Department of Mathematical Sciences, KAIST, South Korea}
\address[S\ae ther, Telle]{Department of Informatics, University of Bergen, Norway}
\email{jjisu@kaist.ac.kr}
\email{sigve.sether@ii.uib.no}
\email{telle@ii.uib.no}
\thanks{The first author is supported by Basic Science Research
  Program through the National Research Foundation of Korea (NRF)
  funded by  the Ministry of Science, ICT \& Future Planning
  (2011-0011653).}
\date{\today}

\begin{abstract}
 We give alternative definitions for maximum matching width, {\sl e.g.} a graph $G$ has $\mmw(G) \leq k$ if and only if it is a subgraph of a chordal graph $H$ and for every maximal clique $X$ of $H$
there exists $A,B,C \subseteq X$ with $A \cup B \cup C=X$ and $|A|,|B|,|C| \leq k$
such that any subset of $X$ that is a minimal separator of $H$ is a subset of
either $A, B$ or $C$.  Treewidth and branchwidth have
alternative definitions through intersections of subtrees, where treewidth
focuses on nodes and branchwidth focuses on edges. We show that mm-width
combines both aspects,
focusing on nodes and on edges.
Based on this we prove that given a graph $G$ and a branch decomposition of mm-width $k$ we can solve Dominating Set in time $O^*({8^k})$, thereby beating $O^*(3^{\tw(G)})$ whenever $\tw(G) > \log_3{8} \times k \approx 1.893  k$. Note that $\mmw(G) \leq   \tw(G)+1 \leq 3 \mmw(G)$ and these inequalities are tight. Given only the graph $G$ and using the best known  algorithms to find decompositions, maximum matching width will be better for solving Dominating Set whenever $\tw(G) >  1.549 \times \mmw(G)$.
\end{abstract}
\maketitle

\section{Introduction}

The treewidth $\tw(G)$ and branchwidth $\bw(G)$ of a graph $G$ are connectivity parameters of importance in algorithm design. By dynamic programming along the associated tree decomposition or branch decomposition one can solve many graph optimization problems in time linear in the graph size and exponential in the parameter. 
For any graph $G$, its treewidth and branchwidth is related by $\bw(G) \leq \tw(G)+1 \leq \frac{3}{2}\bw(G)$ \cite{RS1991}.
The two parameters are thus equivalent with respect to fixed parameter tractability (FPT), with a problem being FPT parameterized by treewidth if and only if it is FPT parameterized by branchwidth.
For some of these problems the best known FPT algorithms are optimal, up to some complexity theoretic assumption. For example, Minimum Dominating Set Problem can be solved in time $O^*(3^{\tw(G)})$ when given a decomposition of treewidth $\tw(G)$ \cite{RBR2009} but not in time
$O^*((3-\epsilon)^{\tw(G)})$ for any $\epsilon >0$ unless the Strong Exponential Time Hypothesis (SETH) fails \cite{DDS2011}.

Recently, a graph parameter equivalent to treewidth and branchwidth was introduced, the maximum matching width (or mm-width) $\mmw(G)$, defined by a branch decomposition over the vertex set of the graph $G$, using the symmetric submodular cut function obtained by taking the size of a maximum matching of the bipartite graph crossing the cut (by K{\"o}nig's Theorem equivalent to minimum vertex cover) \cite{Vatshelle2012}. For any graph $G$ we have $\mmw(G) \leq  \bw(G) \leq \tw(G)+1 \leq 3\mmw(G)$ and these inequalities are tight, for example any balanced decomposition tree will show that $\mmw(K_n)=\lceil \frac{n}{3} \rceil$.

In this paper we show that given a branch decomposition over the vertex set of mm-width $k$ we can solve Dominating Set in time $O^*(8^{k})$. This runtime beats the $O^*(3^{\tw(G)})$ algorithm for treewidth \cite{RBR2009} whenever $\tw(G) > \log_3{8} \times k \approx 1.893 k$. 
If we assume only $G$ as input, then 
since mm-width has a submodular cut function \cite{ST2014} we can approximate
mm-width to within a factor $3\mmw(G)+1$ in $O^*(2^{3\mmw(G)})$ time using the
generic algorithm of \cite{OS2004}, giving a total runtime for solving
dominating set of $O^*(2^{9\mmw(G)})$. For treewidth we can in
$O^*(2^{3.7\tw(G)})$ time \cite{AE2010} get an approximation to within a factor
$(3 + 2/3)\tw(G)$ giving a total runtime for solving dominating set of
$O^*(3^{3.666\tw(G)})$\footnote{%
Note that there is also a $O^*(c^{\tw(G)})$ time $3$-approximation of
treewidth~\cite{bodlaender2013c}, but the $c$ is so large that the approximation
alone has a bigger exponential part than the entire Dominating Set algorithm when
using the $3.666$-approximation.}. This implies that on input $G$, using maximum matching width gives better exponential factors whenever $\tw(G) >  1.549\mmw(G)$.

Our results are based on a new characterization of graphs of mm-width at most $k$, as intersection graphs of subtrees of a tree. It can be formulated as follows, encompassing analogous formulations for all three parameters mm-width (respectively treewidth, respectively branchwidth):

\bigskip

For any $k \geq 2$ a graph $G$ on vertices $v_1,v_2,...,v_n$ has 
$\mmw(G) \leq k$ (resp. $\tw(G) \leq k-1$, resp. $\bw(G) \leq k$) 
if and only if there is a tree $T$ of max degree at most $3$ with nontrivial subtrees $T_1,T_2,...,T_n$ such that if $v_iv_j \in E(G)$ then subtrees $T_i$ and $T_j$ have at least one
node (resp. node, resp. edge)
of $T$ in common and for each 
edge (resp. node, resp. edge)
of $T$ there are at most $k$ subtrees using it.

\bigskip 
Thus, while treewidth has a focus on nodes and branchwidth a focus on edges,
mm-width combines the aspects of both.
We also arrive at the following alternative characterization: a graph $G$ has $\mmw(G) \leq k$ if and only if it is a subgraph of a chordal graph $H$ and for every maximal clique $X$ of $H$
there exists $A,B,C \subseteq X$ with $A \cup B \cup C=X$ and $|A|,|B|,|C| \leq k$
such that any subset of $X$ that is a minimal separator of $H$ is a subset of either $A, B$ or $C$.
In fact, 
using techniques introduced by Bodlaender and Kloks in \cite{BK1996} these new characterizations will also allow us to compute a branch decomposition of optimal mm-width in FPT time \cite{JST2015}.
In section 2 we give definitions. In section 3 we define unique minimum vertex covers for any bipartite graph, show some monotonicity properties of these, and use this to give the new characterizations of mm-width. In section 4 we give the dynamic programming algorithm for dominating set. We end in section 5 with some discussions.

\section{Definitions}

For a simple and loopless graph $G=(V,E)$ and its vertex $v$, let $N(v)$ be the set of all vertices adjacent to $v$ in $G$, and $N[v]=N(v)\cup\{v\}$. For a subset $S$ of $V(G)$, let $N(S)$ be the set of all vertices that are not in $S$ but are  adjacent to some vertex of $S$ in $G$, and $N[S]=N(S)\cup S$. 

A \emph{tree decomposition} of a graph $G$ is a pair $(T,\{X_t\}_{t\in V(T)})$ consisting of a tree $T$ and a family $\{X_t\}_{t\in V(T)}$ of vertex sets $X_t\subseteq V(G)$, called \emph{bags}, satisfying the following three conditions:
\begin{enumerate}
\item each vertex of $G$ is in at least one bag,
\item for each edge $uv$ of $G$, there exists a bag that contains both $u$ and $v$, and
\item for vertices $t_1,t_2,t_3$ of $T$, if $t_2$ is on the path from $t_1$ to $t_3$, then $X_{t_1}\cap X_{t_3}\subseteq X_{t_2}$.
\end{enumerate}
The \emph{width} of a tree decomposition $(T,\{X_t\}_{t\in V(T)})$ is $\max_{t\in V(T)} \abs{X_t}-1$.
The \emph{treewidth} of $G$, denoted by $\tw(G)$, is the minimum width over all possible tree decompositions of $G$.

A \emph{branch decomposition} over $X$, for some set of elements $X$, is a pair
$(T, \delta)$, where $T$ is a tree over vertices of degree at most $3$, and
$\delta$ is a bijection from the leaves of $T$ to the elements in $X$.
Any edge $ab$ disconnects $T$ into two subtrees $T_a$ and $T_b$. Likewise, any
edge $ab$ partitions the elements of $X$ into two parts $A$ and $B$, namely
the elements mapped by $\delta$ from the leaves 
of $T_a$, and of $T_b$, respectively. An edge $ab \in E(T)$ is said to \emph{induce} the
partition $(A, B)$. 

A \emph{rooted branch decomposition} is a
branch decomposition $(T, \delta)$ where we subdivide an edge of $T$ and make the new vertex the root $r$.
In a rooted branch decomposition, for an internal vertex $v
\in V(T)$, we denote by $\delta(v)$ the union of $\delta(l)$ for all leaves of
$l$ having $v$ as its ancestor.

Given a symmetric ($f(A) = f( \compl A)$) function $f : 2^{X} \to \mathbb{R}$,
using branch decompositions over $X$, we get a nice way of defining width
parameters: For a branch decomposition $(T, \delta)$ and edge $e \in T$, we
define the \emph{$f$-value} of the edge $e$ to be the value $f(A) = f(B)$ where
$A$ and $B$ are the two parts of the partition induced by $e$ in $(T, \delta)$,
denoted $f(e)$. We define the \emph{$f$-width} of branch decomposition $(T,
\delta)$ to be the maximum $f$-value over all edges of $T$, denoted $f(T,
\delta)$: $\max_{e \in T}\{f\text{-value of } e\}$. For set $X$ of elements, we
define the \emph{$f$-width} of $X$ to be the minimum $f$-width over all branch
decompositions over $X$.  If $\lvert X\rvert\le 1$, then $X$ admits no branch
decomposition and we define its $f$-width to be $f(\emptyset)$.

For a graph $G$ and a subset $S\subseteq E(G)$, 
the \emph{branchwidth} $\bw(G)$ of $G$ is 
the $f$-width of $E(G)$
where $f: 2^{E(G)} \to \mathbb{R}$ is a function 
such that 
$f(S)$ is the number of vertices that are incident with an edge in $S$ as well as
an edge in $E(G)\setminus S$.

The \emph{Maximum Matching-width} of a graph $G$, mm-width in short, is a width parameter
defined through branch decompositions over $V(G)$ and the cardinality of
matchings. For a subset $S \subseteq V(G)$, the Maximum Matching-value is defined
to be the size of a maximum matching in $G[S, V(G) \setminus S]$, denoted
$\mm(S)$. The mm-width of a graph $G$, denoted $\mmw(G)$, is the $f$-width of
$V(G)$ for $f = \mm$.

\section{Subtrees of a tree representation for mm-width}

\subsection{K\"onig covers}
In this subsection, we will define canonical minimum vertex covers for any bipartite graph.
Our starting point is a well-known result in graph theory.
%
%------- Koenigs Theorem: ----------------------------
    
    \begin{theorem}[K\"onig's Theorem~\cite{Koenig1931}] \label{thm:Koenig}
      Given a bipartite graph $G$, for any maximum matching $M$ and minimum
      vertex cover $C$ of $G$, the number of edges in $M$ is the same as the
      number of vertices in $C$; $|M| = |C|$.
    \end{theorem}

%-----------------------------------------------------
%
%
Let $(A,B)$ be the vertex partition of $G$.
This statement can be proved in multiple ways. 
The harder
direction, that a maximum matching is never smaller than a minimum vertex cover,
does not hold for general graphs, and is usually proven by taking a maximum matching $M$ and constructing a vertex
cover $C$ having size exactly $|M|$, as follows:
    \begin{quote}
      For each edge $ab \in M$ (where $a \in A$, and $b \in B$), if $ab$ is part
      of an alternating path starting in an unsaturated vertex of $A$, then put
      $b$ into $C$, otherwise put $a$ into $C$.
    \end{quote}
For a proof that $C$ indeed is a minimum vertex cover of $G$,
see e.g.~\cite{Diestel2010}.  We will call the vertex cover $C$ constructed by the
above procedure the $A$-\emph{K\"onig cover} of $G$. A $B$-K\"onig cover of $G$
is constructed similarly by changing the roles of $A$ and $B$ (see
Figure~\ref{fig:mm:kvc}).
%
%   ------ FIGURE ----------- 
    \begin{figure}[ht!]
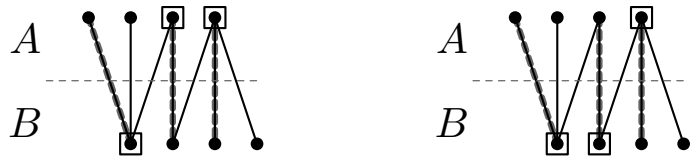

      \center
      \includegraphics[scale=0.8]{KVC_1}
      \hspace{2cm}
      \includegraphics[scale=0.8]{KVC_2}
      \caption{$A$-K\"onig cover and $B$-K\"onig cover.}
      \label{fig:mm:kvc}
    \end{figure}
%   ------------------------- 
%
Lemma~\ref{chp:mm:lemma:unique_koenig} below shows that the $A$-K\"onig cover will, on the $A$-side consist of the $A$-vertices in the union over all minimum vertex covers, and on the $B$-side will consist of the $B$-vertices in the intersection over all minimum vertex covers.

    \begin{lemma}\label{chp:mm:lemma:unique_koenig}
      For a bipartite graph $G = (A \cup B, E)$ and minimum vertex cover $C$ of
      $G$, the set $C$ is the $A$-K\"onig cover of $G$ if and only if for any
      minimum vertex cover $C'$ of $G$ we have $A \cap C' \subseteq A \cap C$,
      and $B \cap C' \supseteq B \cap C$.
    \end{lemma}

\begin{proof} 
  Let $M$ be a maximum matching of $G$, and $C^*$ the $A$-K\"onig cover of $G$
  constructed from $M$. 
Since both $C^*$ and $C$ are minimum vertex covers, by showing that for
  any minimum vertex cover $C'$ of $G$ we have $A \cap C' \subseteq A \cap C^*$,
  and $B \cap C' \supseteq B \cap C^*$, as a consequence will also show that $C'
  = C^*$ if and only if for all minimum vertex covers $C'$ of $G$ we have $A
  \cap C' \subseteq A \cap C$ and $B \cap C' \supseteq B \cap C$.  So this is
  precisely what we will do.

  Let $C'$ be any minimum vertex cover, and $b$ any vertex in $C^* \cap B$.  We
  will show that $b \in C'$, and from that conclude $B \cap C' \supseteq B \cap
  C^*$.  As $b \in C^*$ there must be some alternating path from $b$ to an
  unsaturated vertex $u \in A$.  The vertices $b$ and $u$ are on different sides
  of the bipartite graph, so the alternating path $P$ between $u$ and $b$ must
  be of some odd length $2k+1$. From Theorem~\ref{thm:Koenig}, we deduce
  that one and only one endpoint of each edge in $M$ must be in $C'$.  As each
  vertex in $V(P)$ is incident with at most two edges of $P$, and all edges of
  $P$ must be covered by $C'$, we need at least $\lceil (2k + 1) / 2 \rceil  = k
  + 1$ of the vertices in $V(P)$ to be in $C'$. However, the vertices of $V(P) -
  b$ are incident with only $k$ edges of $M$. Therefore at most $k$ of the
  vertices $V(P) - b$ can be in $C'$.  In order to have at least $k+1$ vertices
  from $V(P)$ in $C'$ we thus must have $b \in C'$.

  We now show that $C' \cap A \subseteq C^*$ by showing that $a \in C^*$ if $a
  \in A \cap C'$.  Let $E^*$ and $E'$ be the edges of $G$ not covered by $C^*
  \cap B$ and $C' \cap B$, respectively. Since $C^* \cap B \subseteq C' \cap B$,
  the set $E^*$ must contain all the edges of $E'$.  As $C'$ is a minimum vertex
  cover, and all edges other than $E'$ are covered by $C' \cap B$, a vertex $a$
  of $A$ is in $C'$ only if it covers an edge $e \in E'$. As $E' \subseteq E^*$,
  we have $e \in E^*$, and hence $C^*$ must also cover $e$ by a vertex in $A$.
  As $G$ is bipartite, the only vertex from $A$ that covers $e$ is $a$, and we
  can conclude that $a \in C^*$.  
\end{proof} % }}}

The following lemma establishes an important monotonicity property for $A$-K\"onig covers.

     \begin{lemma} \label{mm:lemma:legal}
        Given a graph $G$ and tripartition $(A,B,X)$ of the vertices $V(G)$, the
        following two properties holds for the $A$-K\"onig cover $C_A$ of
        $G[A,B \cup X]$ and any minimum vertex cover $C$ of $G[A \cup X, B]$.
        \begin{enumerate}
          \item $A \cap C \subseteq A \cap C_A$
          \item $B \cap C \supseteq B \cap C_A$.
        \end{enumerate}
     \end{lemma}

\begin{proof} % {{{ 
 To prove this, we will show that it holds for $X = \{x\}$, and then by
 transitivity of the subset relation and that a K\"onig cover is also a minimum
 vertex cover, it must hold also when $X$ is any subset
 of $V(G)$. 

 Let $A' = A + x$ and $B' = B + x$, and let $C'$ be the $A$-K\"onig cover of the graph $G[A, B]$ (be aware that this
 graph has one less vertex than $G$). 
 We will break the proof into four parts, namely
       $A \cap C \subseteq A \cap C' $, 
       $A \cap C'  \subseteq A \cap C_A$,
       $B \cap C_A \subseteq B \cap C' $, and
       $B \cap C'  \subseteq B \cap C$.
 Again, by transitivity of the subset relation, this will be sufficient for our
 proof. We now look at each part separately.
 
 \underline{$A \cap C \subseteq A \cap C'$}: 
 Two cases: $|C|=|C'|$ and $|C| > |C'|$. We do the latter first.
 This means that $C' \cup \{x\}$ must be a minimum vertex
 cover of $G[A', B]$. Therefore the $A'$-K\"onig cover $C^*$ of $G[A',B']$ must
 contain $(C' \cup \{x\}) \cap A'$. This means that $C^*$ is a minimum vertex
 cover of $G[A,B]$, and by $C'$ being the $A$-K\"onig cover of $G[A,B]$, we have
 from Lemma~\ref{chp:mm:lemma:unique_koenig} that $C' \cap A \supseteq C^* \cap
 A$. And since $C^*$ is a $A'$-K\"onig cover of $G[A',B]$ we have $C' \cap A'
 \supseteq C \cap A'$ and can conclude that $C' \cap A \supseteq A \cap C$.
 Now assume that the two vertex covers are of equal size. Clearly $x \not\in
 C$, as then $C - x$ is a smaller vertex cover of $G[A,B]$ than $C'$, so $x$
 is not in $C$. This means that $C$ is a minimum vertex cover of $G[A,B]$,
 so all vertices in $A \cap C$ must be in $C'$ by Lemma~\ref{chp:mm:lemma:unique_koenig}.

 \underline{$A \cap C'  \subseteq A \cap C_A$}:
 Suppose $C'$ is smaller than $C_A$. This means $C' + x$ is a minimum vertex
 cover of $G[A,B']$, and hence $(C' + x) \cap A \subseteq C_A \cap A$ by Lemma~\ref{chp:mm:lemma:unique_koenig}. On the other hand, if $C'$ is of the
 same size as $C_A$. Then $C_A$ is a minimum vertex cover of $G[A, B]$, and so $x
 \not\in C_A$.  This means $C_A \cap N(x) \cap A \subseteq C_A \cap A$. And as $C_A$ is
 a minimum vertex cover of $G[A,B]$, we know from
 Lemma~\ref{chp:mm:lemma:unique_koenig} that $C_A \cap N(x) \cap A \subseteq C'$. In
 particular, this means $C'$ covers all the edges of $G[A,B']$ not in $G[A,B]$,
 which means that $C'$ is also a minimum vertex cover of $G[A,B']$. This latter
 observation means that $C' \cap A \subseteq C_A \cap A$ from Lemma~\ref{chp:mm:lemma:unique_koenig}.  
 
 \underline{$B \cap C_A \subseteq B \cap C' $}:
 Suppose $C'$ is smaller than $C_A$. This means $C' + x$ is a minimum vertex
 cover of $G[A,B']$, and thus $B' \cap (C' + x) \supseteq B' \cap C_A$. Which
 implies that $B \cap C' \supseteq B \cap C_A$.
 Now assume that $C'$ is of the same size as $C_A$. This means $C_A$ is a
 minimum vertex cover of $G[A,B]$ and $x \not\in C_A$. Furthermore, this means
 $N(x) \cap A \subseteq C_A \cap A \subseteq C' \cap A$ by
 Lemma~\ref{chp:mm:lemma:unique_koenig} and we conclude that $C'$ is a minimum
 vertex cover of $G[A,B']$. By Lemma~\ref{chp:mm:lemma:unique_koenig}, this
 means $B' \cap C_A \subseteq B' \cap C'$ and in particular $B \cap C_A
 \subseteq B \cap C'$.

 \underline{$B \cap C'  \supseteq B \cap C$}:
 Suppose $C'$ is smaller than $C$. This means $C' + x$ is a minimum vertex
 cover of $G[A,B']$, and hence by Lemma~\ref{chp:mm:lemma:unique_koenig} we have 
 $B' \cap (C' + x) \subseteq B' \cap C_2$, which implies 
 $B \cap C' \subseteq B \cap C_2$.
 Now suppose $C'$ is of the same size as $C$. This means that $C$ is a
 minimum vertex cover of $G[A,B]$, and hence we immediately get $C \cap B
 \supseteq C' \cap B$ by Lemma~\ref{chp:mm:lemma:unique_koenig}.
 
 This completes the proof, as we by transitivity of the subset relation have that
 $C_A \cap B \subseteq C \cap B$, and $C \cap A \subseteq C_A \cap A$.
\end{proof} % }}} 

We are now ready to prove an important connectedness property of K\"onig covers that arise from cuts of a given branch decomposition.

\begin{lemma} \label{mm:lemma:monotoneunion}
 Given a connected graph $G$ and rooted branch decomposition $(T, \delta)$ over $V(G)$,
 for any node $v$ in $T$, where $\mathcal{C}$ are the descendants of $v$ and 
 $C_u$ means the $\delta(u)$-K\"onig cover of $G[\compl{\delta(u)},
 \delta(u)]$, we have that 
 \[
   \left(\bigcup_{x \in V(T) \setminus \mathcal{C}} C_x \right) \cap 
   \left(\bigcup_{x \in \mathcal{C}} C_x \right) 
   \subseteq C_v \enspace .
 \]
\end{lemma}

\begin{proof} % {{{
 First notice for all $x \in \mathcal{C}$, since $C_x$
 is a $\delta(x)$-K\"onig cover and $C_v$ a minimum vertex cover, from
 Lemma~\ref{mm:lemma:legal} we have that $C_x \cap
 \compl{\delta(x)} \subseteq C_v \cap \compl{\delta(x)}$.  In
 particular, since $\delta(x) \subseteq \delta(v)$, we have that $C_x \setminus
 \delta(v) \subseteq C_v \setminus \delta(v) \subseteq C_v$.
 Since each vertex of $V(G)$ is either in $\delta(v)$ or not in $\delta(v)$, by showing that also for all $x \in
 (V(T) \setminus \mathcal{C})$ we have $C_x \cap \delta(v) \subseteq C_v$ we
 can conclude that the lemma holds: 
 For all $x \in V(T) \setminus \mathcal{C}$ either $\delta(v) \subseteq \delta(x)$
 (when $x$ is an ancestor of $v$) or $\delta(v) \subseteq \compl{\delta(x)}$
 (when $x$ is neither a descendant of $v$ nor an ancestor of $v$), in either case, we
 can apply the $\delta(v)$-K\"onig cover $C_v$ of $G[\delta(v),
 \compl{\delta(v)}]$ and the minimum vertex cover $C_x$ of $G[\delta(x),
 \compl{\delta(x)}]$ to Lemma~\ref{mm:lemma:legal} and see that $C_x \cap
 \delta(v) \subseteq C_v \cap \delta(v) \subseteq C_v$. 
\end{proof} % }}}

\subsection{The new characterization of mmw}

We say a graph is \emph{nontrivial} if it has an edge.

\begin{theorem}\label{thm:newcharac}
A nontrivial graph $G=(V,E)$ has $\mmw(G) \leq k$ if and only if there is a tree $T$ of max degree at most $3$ and for each vertex $u \in V$ a nontrivial subtree $T_u$ of $T$ such that i) if $uv \in E$ then the subtrees $T_u$ and $T_v$ have at least one vertex of $T$ in common, and ii) for every edge of $T$ there are at most $k$ subtrees using this edge.
\end{theorem}

\begin{proof}
Forward direction:
Let $(T,\delta)$ be a rooted branch decomposition over $V$ having mm-width at
most $k$, and assume $G$ has no isolated vertices.
For each edge $e=uv$ of $T$, with $u$ a child of $v$, assign the $\delta(u)$-K{\"o}nig cover $C_u$ of $G[\delta(u),V \setminus \delta(u)]$ to the edge $uv$.
For each vertex $x$ of $G$, define the set of edges of $T$ whose K{\"o}nig cover contains $x$ and let $T_x$ be the sub-forest of $T$ induced by these edges. Using Lemma~\ref{mm:lemma:monotoneunion} we first show that $T_x$ is a connected forest and thus a subtree of $T$.
Consider edge $e=uv$ of $T$. Let $p$ be the lowest common ancestor of $u$ and $v$.
For every vertex $w$ on the path from $p$ to $u$ and on the path from $p$ to $v$, except $p$,
we know that exactly one of $u,v$ is a descendant of $w$.
By Lemma~\ref{mm:lemma:monotoneunion}, $(C_u \cap C_v) \subseteq C_w$.
It means that if a vertex $x$ of $G$ is in both $C_u$ and $C_v$ then it is also in $C_w$,
which implies that $T_x$ is connected.

Now, since the branch decomposition has mm-width at most $k$ part i) in the statement of the Theorem holds.
For an arbitrary edge $ab$ of $G$, consider any edge $e$ of $T$ on the path from $\delta^{-1}(a)$ to $\delta^{-1}(b)$ and the partition $(A,B)$ induced by $e$ where $a\in A$, $b\in B$. Then the K{\"o}nig cover of $e$ must contain one of $a$ and $b$, and thus, ii) holds as well. 
Finally, $T_x$ is nontrivial because the edge of $T$ incident with a leaf $\delta^{-1}(x)$ assigns the K{\"o}nig cover $\{x\}$.
If $G$ has isolated isolated vertices, $T_x$ is not nontrivial for isolated
vertex $x$. We fix this by setting $T_x$ to consist exactly of the edge incident
with $\delta^{-1}(x)$, for any isolated vertex $x$ of $G$.

Backward direction: For each given subtree $\{T_u\}_{u\in V}$ of $T$, choose an edge in $T_u$ (it is also in $T$) and append in the tree $T$ a leaf $\ell_u$, and extend $T_u$ to contain $\ell_u$ and set $\delta(\ell_u)=u$. Exhaustively remove leaves (from both $T$ and the subtrees) that are not mapped by $\delta$. Call the resulting tree $T'$ and subtrees $\{T'_u\}_{u\in V}$. 
Note that subtrees $\{T'_u\}_{u\in V}$ and $T'$ still satisfy i) and ii). 
We claim that $(T',\delta)$ is a branch decomposition of mm-width at most $k$. It is clearly a branch decomposition over $V$, and for any edge $e$ of $T'$, if we choose $S \subseteq V$ to be those $u$ with $T_u$ using this edge $e$, then this will be a vertex cover of the bipartite graph $H$ given by this edge $e$, and of size at most $k$ because 
for an edge $xy$ in $H$, one of $T_x$ and $T_y$ must contain $e$.
\end{proof}

In the introduction we mentioned analogous characterizations of treewidth and branchwidth, for these see e.g. \cite{PT2009}.
 Another alternative characterization is the following.
 
 \begin{corollary}
A graph $G$ has $\mmw(G) \leq k$ if and only if it is a subgraph of a chordal graph $H$ and for every maximal clique $X$ of $H$
there exists $A,B,C \subseteq X$ with $A \cup B \cup C=X$ and $|A|,|B|,|C| \leq k$
such that any subset of $X$ that is a minimal separator of $H$ is a subset of either $A, B$ or $C$.
 \end{corollary}

 We only sketch the proof, which is similar to an alternative characterization of branchwidth given in \cite{PT2009}.
 We say a tree is \emph{ternary} if it has maximum degree at most $3$.
 Note that a graph is chordal if and only if it is an intersection graph of subtrees of a tree~\cite{Gavril1974}.
 In the forward direction, take the chordal graph resulting from the subtrees of ternary tree representation. In the backward direction, take a clique tree of $H$ and make a ternary tree decomposition (which is easily made into a subtrees of ternary tree representation) by for each maximal clique $X$ of degree larger than three making a bag $X$ with three neighboring bags $A,B,C$. If minimal separators $S_1,...,S_q$ subset of $X$ are contained in $A$ make a path extending from bag $A$ of $q$ new bags also containing $A$, with a single bag containing $S_i, 1 \leq i \leq q$ attached to each of them. These ternary subtrees, one for each maximal clique, is then connected together in a tree by the structure of the clique tree, adding an edge between bags of identical minimal separators.

\section{Fast DP for Dominating Set parameterized by mm-width}

For graph $G=(V,E)$ a subset of vertices $S \subseteq V$ is said to {\em dominate} the vertices in $N[S]$, and it is a {\em dominating set}  if $N[S]=V$.
Given a rooted branch decomposition $(T, \delta)$ of $G$ of mm-width $k$, we will in this section give an $O^*(8^k)$ algorithm for computing the size of a Minimum Dominating Set of $G$. This by an algorithm doing dynamic programming along a rooted tree decomposition $(T', \{X_t\}_{t \in V(T')})$ of $G$ that we compute from $(T, \delta)$ as follows.

Given a rooted branch decomposition $(T, \delta)$ of $G$ having mm-width $k$ the proof of Theorem~\ref{thm:newcharac} yields a polynomial-time algorithm (using an algorithm for maximum matching in bipartite graphs) finding
a family $\{T_u\}_{u\in V(G)}$ of nontrivial subtrees of $T$ (note we can assume $T$ is a rooted tree with root of degree two and all other internal vertices of degree three) such that 
i) if $uv \in E(G)$ then the subtrees $T_u$ and $T_v$ have at least one vertex of $T'$ in common, and 
ii) for every edge of $T'$ there are at most $k$ subtrees using this edge.
From this it is easy to construct a rooted tree decomposition $(T',\{X_t\}_{t \in V(T')})$ of $G$, having the properties described in Figure \ref{fig:tree}.
Let $T'$ be a tree with vertex set $A\cup B \cup \{r\}$ 
where $A$ is the set of edges of $T$, 
$B$ is the set of non-root vertices (all of degree-$3$) of $T$,
 and 
$r$ is the root of $T$ and also the root of $T'$.
Two vertices $e,v$ of $T'$ are adjacent if and only if 
$e\in A$ and $v\in B \cup \{r\} $ are incident in $T$.
For a vertex $e\in A$, let $X_e$ be the set of vertices in $G$ such that 
if a subtree $T_w$ uses edge $e$ of $T$, then $w\in X_e$.
For a vertex $v\in B$, let $X_v$ be the set of vertices in $G$ such that
for the three incident edges $e_1,e_2,e_3$ of $v$ in $T$, $X_v=X_{e_1}\cup X_{e_2}\cup X_{e_3}$.
Let $X_r=X_{e_1}\cup X_{e_2}$ if $e_1$ and $e_2$ are incident with $r$ in $T$.
Then $(T',\{X_t\}_{t \in V(T')})$ is a tree decomposition of $G$ with a root $r$, having the properties described in Figure \ref{fig:tree}, which we will use in the dynamic programming.

\bigskip

Let us now define the relevant subproblems for the dynamic programming over this tree decomposition. 
For node $t$ of the tree we denote by $G_t$ the graph induced by the union of $X_u$ where $u$ is a descendant of $t$.
A coloring of a bag $X_t$ is a mapping $f: X_t \rightarrow \{1, 0, *\}$ with the meaning that: all vertices with color 1 are contained in the dominating set of this partial solution in $G_t$, all vertices with color 0 are dominated, while vertices with color * might be dominated, not dominated, or in the dominating set. Thus the only restriction is that a  vertex with color 1 must be a dominator, and a vertex with color 0 must be dominated.
Thus, for any $S \subseteq V(G)$ there is a set $c(S)$ of $3^{|S|}2^{|N(S)|}$ colorings $f: V(G) \rightarrow \{1, 0, *\}$ compatible with taking $S$ as set of dominators, 
with vertices of $S$ colored 1, 0 or $*$, vertices of $N(S)$ colored 0 or $*$, and the remaining vertices colored $*$.

For a coloring $f$ of bag $X_t$ we denote by $T[t,f]$ (and view this as a 'Table' of values) the minimum $|S|$ over all $S \subseteq V(G_t)$ such that there exists $f' \in c(S)$ with $f'|_{X_t}=f$ and $f'|_{V(G_t) \setminus X_t}$ having everywhere the value 0.
In other words, the minimum size of a set $S$ of vertices of $G_t$ that dominate all vertices in $V(G_t) \setminus X_t$, with a coloring $f'$ compatible with taking $S$ as set of dominators, such that $f'$ restricted to $X_t$ gives $f$.  
 If no such set $S$ exists, then $T[t,f]= \infty$.
Note that the size of the minimum dominating set of $G$ is the minimum value over all $T[r,f]$ where $f^{-1}(*)=\emptyset$ at the root $r$.
We initialize the table at a leaf $\ell$, with $X_{\ell}=\{v\}$ as follows. Denote by $f_i$ the coloring from $\{v\}$ to $\{1,0,*\}$ with $f_i(v)=i$ for $i\in \{1,0,*\}$.
Then for a leaf bag $X_{\ell}$, set $T[\ell,f_1]:=1$, $T[\ell,f_0]:=\infty$, $T[\ell,f_{*}]:=0$.

For internal nodes of the tree, instead of separate `Join, Introduce and Forget' operations we will give a single update rule with several stages.
We will be using an Extend-Table subroutine which takes a partially filled table $T[t,\cdot]$ and extends it to table $T'[t,\cdot]$
so the result will adhere to the above definition, ensuring the monotonicity property that $T'[t,f] \leq T'[t,f']$ for any $f$ we can get from $f'$ by changing the color of a vertex from 1 to 0 or $*$, or from 0 to $*$. Extend-Table is implemented as follows:

\begin{enumerate}[(a)]
\item Initialize. For all $f$, if $T[t,f]$ is defined then $T'[t,f]:=T[t,f]$, else $T'[t,f]:=\infty$.

\item Change from 1 to 0. For $q=|X_t|$ down to 1: for any $f$ in $T'[t,f]$ where $|\{v: f(v)=1\}|=q$, for any choice of a single vertex $u \in \{v: f(v)=1\}$ set $f_u(u)=0$ and set $f_u(x)=f(x)$ for $x \neq u$, and update $T'[t,f_u]:=\min \{T'[t,f_u], T'[t,f]\}$. 

\item Change from 0 to $*$. Similarly as in step (b).
\end{enumerate}

Note the transition from color 1 to $*$ will happen by transitivity. The time for Extend-Table is proportional to the number of entries in the tables times $|X_t|$.

%   ------ FIGURE ----------- 
    \begin{figure}[ht!]
      \center
      \includegraphics[scale=0.5]{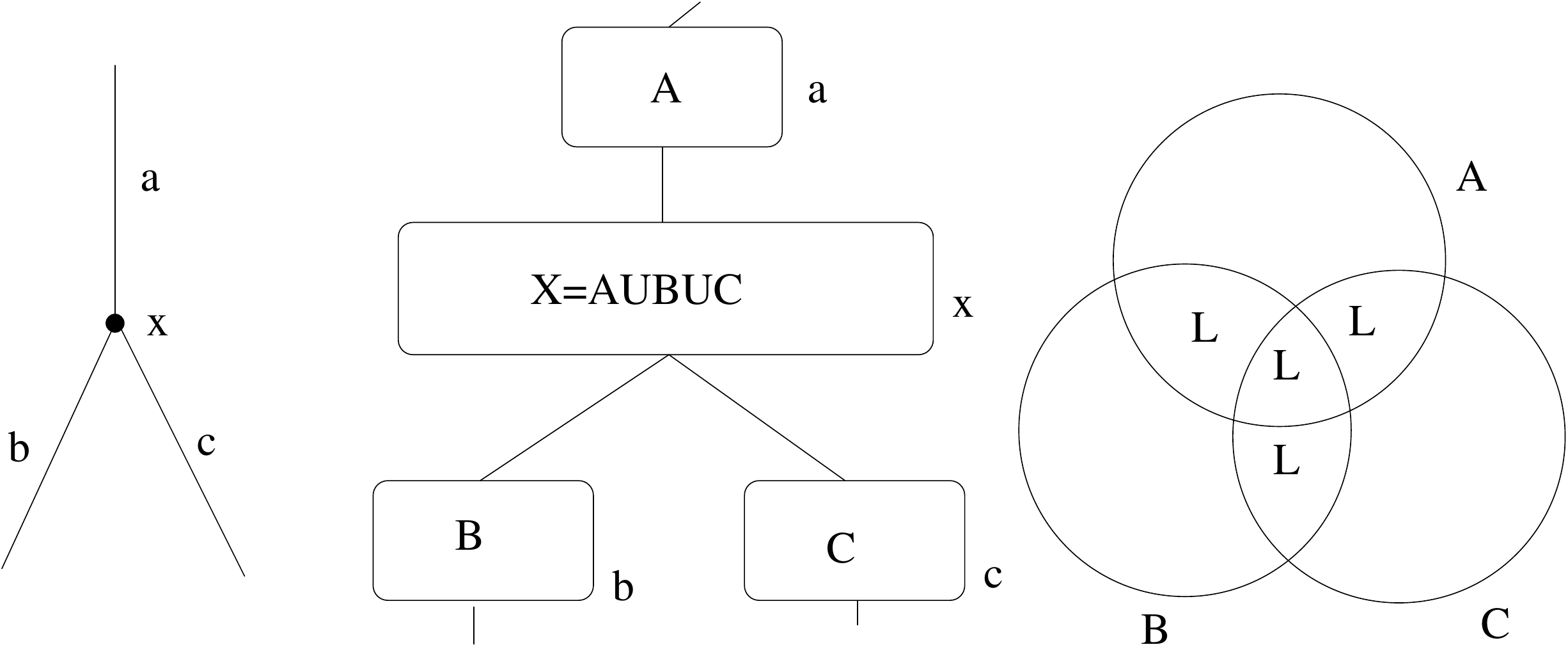}
      \caption{Part of ternary tree used in the subtree representation of $G$ on the left, with node $x$ having three incident edges $a,b,c$, with subtrees of vertices contained in $A,B,C \subseteq V(G)$ using these edges respectively, giving rise to the four bags in the tree decomposition shown in the middle, with constraint $|A|,|B|,|C| \leq k$. }
      \label{fig:tree}
    \end{figure}
%   ------------------------- 

Assume we have the situation in Figure \ref{fig:tree}, corresponding to the bags surrounding any degree-three node $x$ of the tree decomposition. This arises from the branch decomposition (and the subtrees of tree representation) having a node incident to three edges, creating three bags $a,b,c$ containing subsets of vertices $A,B,C$, respectively, each of size at most $k$, and giving rise to the four bags $a,b,c,x$ in the figure, with the latter containing subsets of vertices  $X= A \cup B \cup C$. Let $L= (A \cap B) \cup (A \cap C) \cup (B \cap C)$.
Assume we have already computed $T[b,f]$ and $T[c,f]$ for all $3^{|B|}$ and
$3^{|C|}$ choices of $f$, respectively. We want to compute $T[a,f]$ for all
$3^{|A|}$ choices of $f$, in time $O^*(\max\{3^{|A|}, 3^{|B|}, 3^{|C|}, 3^{|L|} 2^{|X \setminus L|}\})$. 
Note that we will not compute the table $T[x,\cdot]$, as it would have $3^{|X|}$ entries, which is more than the allowed time bound. Instead, we compute a series of tables:

\begin{enumerate}[(1)]
 \item $T^1_b[x,\cdot]$ (and $T^1_c[x,\cdot]$) of size $3^{|B|}$, by for each entry $T[b,f]$ extending the coloring $f$ of $B$ to a unique coloring $f'$ of $X$ based on the neighborhood of the dominators in $f$,
 \item $T^2_b[x,\cdot]$ (and $T^2_c[x,\cdot]$) of size at most $\min (3^{|B|}, 3^{|B \cap L|}2^{|X \setminus (B \cap L)|})$, by changing each coloring $f$ of $X$ to a coloring $f'$ of $X$ where vertices in $B \setminus L$ having color 1 instead are given color 0 (note these vertices have no neighbors in $V(G) \setminus V(G_x)$),
 \item $T^3_b[x,\cdot]$ (and $T^3_c[x,\cdot]$) of size exactly $3^{|B \cap L|}2^{|X \setminus (B \cap L)|}$, with $f^{-1}(1) \subseteq B \cap L$, by running Extend-Table on $T^2_b[x,\cdot]$,
 \item $T^1_{sc}[x,\cdot]$ of size $3^{|L|}2^{|X \setminus L|}$ by subset convolution over parts of $T^3_b[x,\cdot]$  and $T^3_c[x,\cdot]$,
 \item $T^2_{sc}[x,\cdot]$ of size $3^{|L|}2^{|X \setminus L|}$ by running Extend-Table on $T^1_{sc}[x,\cdot]$,
 \item $T[a, \cdot]$ of size $3^{|A|}$ by going over all $3^{|A|}$ colorings of $A$ and minimizing over appropriate entries of $T^2_{sc}[x,\cdot]$.
\end{enumerate}

Note that in Step (4) we use the following:

\begin{theorem} [Fast Subset Convolution \cite{BHKK2007}] \label{fss}
For two functions $g,h:2^V\rightarrow \{-M,\ldots,M\}$, given all the $2^{\abs{V}}$ values of $g$ and $h$ in the input,
all $2^{\abs{V}}$ values of the subset convolution of $g$ and $h$ over the integer min-sum semiring,
i.e. 
$(g * h)(Y)= \min_{Q \cup R=Y \mbox{ and } Q \cap R= \emptyset} g(Q)+h(R)$,
can be computed in time $2^{\abs{V}}\abs{V}^{O(1)}\cdot O(M \log M \log\log M)$.
\end{theorem}

Let us now give the details of the first three steps:

\begin{enumerate}[(1)]
 \item Compute $T^1_b[x,\cdot]$. 
In any order, go through all $f: B \rightarrow \{1,0,*\}$ and compute $f': B \cup A \cup C \rightarrow \{1,0,*\}$ by 
 $$f'(v)= \left\{ \begin{array}{ll}
                   f(x) & \mbox{if $v \in B$}\\
                   0 & \mbox{if $v \not \in B$ and $\exists u \in B: f(u)=1 \wedge uv \in E(G)$}\\
                   * & \mbox{otherwise}\\
                  \end{array} \right. $$
and set $T^1_b[x,f']:=T[b,f]$. 
\item Compute $T^2_b[x,\cdot]$. First, initialize $T^2_b[x,f] = \infty$ for all
$f: B \cup A \cup C \to \{1,0,*\}$ where $f^{-1}(1) \subseteq B \cap L$.
In any order, go through all $f: B \cup A \cup C \rightarrow \{1,0,*\}$ such that $T^1_b[x,f]$ was defined in the previous step, 
 and compute $f': B \cup A \cup C \rightarrow \{1,0,*\}$ by 
$$f'(v)= \left\{ \begin{array}{ll}
                    0 & \mbox{if $v \in B \setminus L$ and $f(v)=1$}\\
                  f'(v)=f(v) & \mbox{otherwise}\\
                  \end{array} \right. $$ 
and set $T^2_b[x,f']:=\min \{T^2_b[x,f'], T^1_b[x,f]\}$.
There will be no other entries in $T^2_b[x,\cdot]$. 
\item Compute $T^3_b[x,\cdot]$ by Extend-Table on $T^2_b[x,\cdot]$.
\end{enumerate}

The total time for the above three steps is bounded by $O^*(\max\{3^{|B|}, 3^{|B \cap L|}2^{|X \setminus (B \cap L)|}\})$.
Note that $T^3_b[x,f]$ is defined for all $f$ where vertices in $B \cap L$ take on values $\{1,0,*\}$ and vertices in 
$X \setminus (B \cap L)$ take on values $\{0,*\}$.
The value of $T^3_b[x,f]$ will be the minimum $|S|$ over all $S \subseteq V(G_b)$ such that there exists $f' \in c(S)$ with $f'|_X=f$ and $f'|_{V(G_b) \setminus X}$ having everywhere the value 0.
Note the slight difference from the standard definition, namely that even though the coloring $f$ is defined on $X$, the dominators only come from $V(G_b)$, and not from $V(G_x)$.
The table $T^3_c[x,\cdot]$ is computed in a similar way, with the colorings again defined on $X$ but with the dominators now coming from $V(G_c)$.

When computing a Join of these two tables, we want dominators to come from $V(G_b) \cup V(G_c)$.
Because of the monotonicity property that holds for these two tables, we can compute their Join $T^1_{sc}[x,f]$ for any $f$ where vertices in $L$ take on values $\{1,0,*\}$ and vertices in 
$X \setminus L$ take on values $\{0,*\}$, by combining colorings as follows: 

$$T^1_{sc}[x,f]= \min_{f_b, f_c} (T^3_b[x,f_b]+T^3_c[x,f_c])-|f^{-1}(1) \cap B \cap C|$$

where $f_b,f_c$ satisfy:

\begin{itemize}
 \item $f(v)=0$ if and only if $(f_b(v),f_c(v)) \in \{(0,*),(*,0)\}$
 \item $f(v)=*$ if and only if $f_b(v)=f_c(v)=*$
 \item $f(v)=1$ if and only if $v \in B \cap C$ and $f_b(v)=f_c(v)=1$, or $v \in B \setminus C$ and $(f_b(v), f_c(v))=(1,*)$,
 or $v \in C \setminus B$ and $(f_b(v), f_c(v))=(*,1)$.
\end{itemize}

This means that we can apply subset convolution to compute a table $T^1_{sc}[x,f]$ on $3^{|L|}2^{|X \setminus L|}$ entries based on $T^3_b[x,f]$ and $T^3_c[x,f]$. Note that $(B \cap L) \cup (C \cap L)=L$. For this step we follow the description in \cite[Section 11.1.2]{PCbook2015}. 
Fix a set $D \subseteq L$ to be the dominating vertices. Let $F_D$ denote the set of $2^{|X \setminus D|}$ functions $f: X \rightarrow \{1,0,*\}$ such that $f^{-1}(1)=D$, i.e. with vertices in $X \setminus D$ mapping in all possible ways to $\{0,*\}$. For each $D \subseteq L$ we will by subset convolution compute the values of $T^1_{sc}[x, f]$ for all $f \in F_D$. 

We represent every $f \in F_D$ by the set $S=f^{-1}(0)$ and define $b_S: X \rightarrow \{1,0,*\}$ such that $b_S(x)=1$ if $x \in D \cap B$, $b_S(x)=0$ if $x \in S$, $b_S(x)=*$ otherwise. Similarly, define $c_S: X \rightarrow \{1,0,*\}$ such that $c_S(x)=1$ if $x \in D \cap C$, $c_S(x)=0$ if $x \in S$, $c_S(x)=*$ otherwise.
Then, as explained previously, for every $f \in F_D$ we want to compute
$$T^1_{sc}[x,f]= \min_{Q \cup R=f^{-1}(0) \mbox{ and } Q \cap R= \emptyset} (T^3_b[x,b_Q]+T^3_c[x,c_R])-|f^{-1}(1) \cap B \cap C|.$$
Define functions $T_b:2^{X \setminus D} \rightarrow \nat$ such that for every $S \subseteq X \setminus D$ we have $T_b(S)=T^3_b[x,b_S]$. Likewise, define functions $T_c:2^{X \setminus D} \rightarrow \nat$ such that for every $S \subseteq X \setminus D$ we have $T_c(S)=T^3_c[x,c_S]$. Also, define $a_S: X \rightarrow \{1,0,*\}$ such that $a_S(x)=1$ if $x \in D$, $a_S(x)=0$ if $x \in S$, $a_S(x)=*$ otherwise. We then compute for every $S \subseteq X \setminus D$,
$$T^1_{sc}[x,a_S]:=(T_b * T_c)(S)-|f^{-1}(1) \cap B \cap C|$$
where the subset convolution is over the mini-sum semiring. 

\bigskip

(4) In Step (4), by Fast Subset Convolution, Theorem \ref{fss}, we compute
$T^1_{sc}[x,a_S]$, for all $a_S$ defined by all $f \in F_D$, in $O^*(2^{|X
\setminus D|})$ time each. 
For all such subsets $D \subseteq L$ we get the time 
$$\sum_{D \subseteq L} 2^{|X \setminus D|} = 
 \sum_{D \subseteq L} 2^{|X \setminus L|}2^{|L \setminus D|}=  
 2^{|X \setminus L|} \sum_{D \subseteq L} 2^{|L \setminus D|} =
   2^{|X \setminus L |}3^{|L|}.$$

(5) In Step (5), we need to run Extend-Table on $T^1_{sc}[x, \cdot]$ to get the table $T^2_{sc}[x, \cdot]$.
This since the subset convolution was computed for each fixed set of dominators so the monotonicity property of the table may not hold. 
Note that the value of $T^2_{sc}[x, f]$ will be the minimum $|S|$ over all $S \subseteq V(G_b) \cup V(G_c)$ such that there exists $f' \in c(S)$ with $f'|_X=f$ and $f'|_{(V(G_b) \cup V(G_c)) \setminus X}$ having everywhere the value 0.

\medskip

(6) In Step (6), we will for each  $f:A \rightarrow \{1,0,*\}$ compute $f': B \cup A \cup C \rightarrow \{1,0,*\}$ by 
 $$f'(v)= \left\{ \begin{array}{ll}
                   1 & \mbox{if $v \in A \cap L$ and $f(v)=1$}\\
                   0 & \mbox{if $v \in A$ and $f(v)=0$ and $N(v) \cap f^{-1}(1)=\emptyset$}\\ 
                   0 & \mbox{if $v \not \in A$ and $N(v) \cap f^{-1}(1)=\emptyset$}\\
                   * & \mbox{otherwise}\\
                  \end{array} \right. $$
and set $T[a,f]:=T^2_{sc}[x,f']+|f^{-1}(1) \cap (A \setminus L)|$. 

Note that when we iterate over all choices of $f:A \rightarrow \{1,0,*\}$, the
vertices colored 0 (in addition to all vertices of $X \setminus A$) must either
be dominated by the vertices in $f^{-1}(1)$ or by vertices in $X \setminus V_a$.
As we know precisely what vertices of $f^{-1}(0)$ are dominated by $f^{-1}(1)$,
we know the rest must be dominated from vertices of $X \setminus V_a$, and
therefore we look in $T_{sc}[x,f']$ at an index $f'$ which colors the rest of
$f^{-1}(0)$ by 0. We can also observe that it is not important for us whether or
not $f^{-1}(0)$ contains all neighbours of $f^{-1}(1)$, since we are iterating
over all choices of $f$ - also those where $f^{-1}(0)$ contains all neighbours
of $f^{-1}(1)$.

The total runtime becomes $O^*(\max\{3^{|A|}, 3^{|B|}, 3^{|C|}, 3^{|L|} 2^{|(A \cup B \cup C) \setminus L|}\})$, with $L= (A \cap B) \cup (A \cap C) \cup (B \cap C)$ and with constraints $|A|, |B|, |C| \leq k$. This runtime is maximum when $L=\emptyset$, giving a runtime of $O^*(2^{3k})$. We thus have the following theorem.

\begin{theorem}
Given a graph $G$ and branch decomposition over its vertex set of mm-width $k$ we can solve Dominating Set in time $O^*(8^{k})$.
\end{theorem}

\section{Discussion}

We have shown that the graph parameter mm-width will for some graphs be better than treewidth for solving Minimum Dominating Set.
The improvement holds whenever $\tw(G) >  1.549 \times \mmw(G)$, if given only the graph as input.
In Figure \ref{twmmw} we list some examples of small graphs having treewidth at least twice as big as mm-width. It could be interesting to explore the relation between treewidth and mm-width for various well-known classes of graphs.
The given algorithmic technique, using fast subset convolution, should extend to any graph problem expressible as a maximization or minimization over $(\sigma, \rho)$-sets, using the techniques introduced for treewidth in \cite{RBR2009}.

% ------ FIGURE ----------- 
    \begin{figure}[ht!]
      \center
      \includegraphics[scale=0.5]{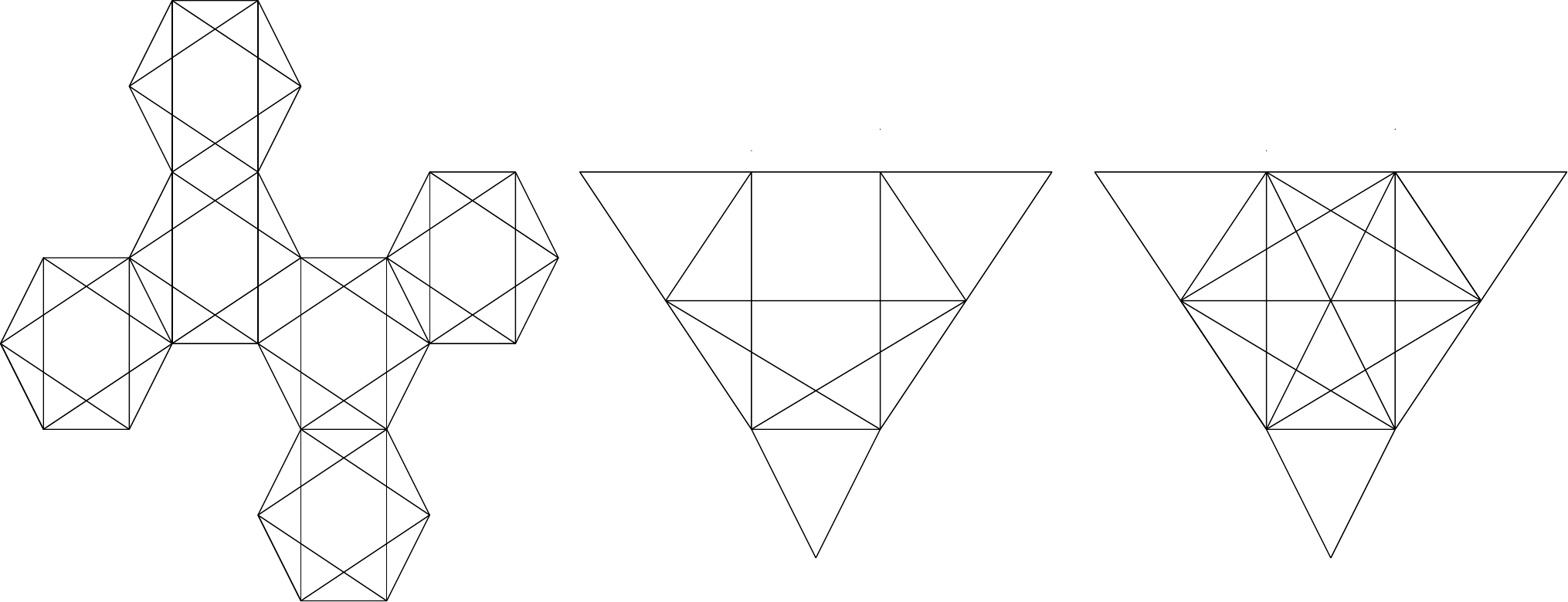}
      \caption{Three graphs of mm-width 2. Left and middle have treewidth 4, and right has treewidth 5. }
      \label{twmmw}
    \end{figure}
%   ------------------------- 

We may also compare with branchwidth. 
Let $\omega$ be the \emph{exponent of matrix multiplication}, which is less than $2.3728639$~\cite{GF2014}.
In 2010, Bodlaender, van Leeuwen, van Rooij, and Vatshelle~\cite{BLRV2010} gave an $O^*(3^{\frac{\omega}{2}k})$ time algorithm 
solving Minimum Dominating Set if an input graph is given with its branch decomposition of width $k$.
This means that given decompositions of $\bw(G)$ and $\mmw(G)$
our algorithm based on mm-width is faster than the algorithm in~\cite{BLRV2010} 
whenever $\bw(G)>\log_3{8}\cdot \frac{2}{\omega} \cdot \mmw(G)>\frac{2\log_3{8}}{2.3728639} \cdot \mmw(G)>1.6 \mmw(G)$.

Taking the subtrees of tree representation for treewidth, branchwidth and maximum matching width mentioned in the Introduction as input, our algorithm for dominating set can be seen as a generic one that works for any of treewidth, branchwidth or maximum matching width of the given representation, and in case of both treewidth and mm-width it will give the best runtime known. 

We gave an alternative definition of mm-width using subtrees of a tree, similar
to alternative definitions of treewidth and branchwidth. We saw that in the
subtrees of a tree representation
treewidth focuses on nodes, branchwidth focuses on edges, and mm-width combines
them both. There is also a fourth way of defining a parameter through these
intersections of subtrees representation; where subtrees $T_u$ and $T_v$ must
share an edge if $uv \in E(G)$ (similar to branchwidth) and the width is defined
by the maximum number of subtrees sharing a single vertex (similar to treewidth).
This parameter will be an upper bound on all the other three parameters, but
might it be that the structure this parameter highlights can be used to
improve the runtime of Dominating Set beyond $O^*(3^{\tw(G)})$ for even more
cases than those shown using mm-width and branchwidth?

\bibliographystyle{abbrv}
\bibliography{bib-width}

\end{document}